\documentclass[11pt]{article}
 \usepackage{geometry}
\usepackage{graphicx}
\usepackage{amssymb}
\usepackage{mathrsfs}
\usepackage{amsmath}
\usepackage[all]{xy}
\usepackage[tight]{subfigure}
\usepackage{xspace}  
\usepackage{multirow}
\usepackage[numbers]{natbib}
\newtheorem{theorem}{Theorem}[section]
\newtheorem{corollary}[theorem]{Corollary}
\newtheorem{lemma}[theorem]{Lemma}
\newtheorem{proposition}[theorem]{Proposition}
\newtheorem{claim}[theorem]{Claim}
\newtheorem{definition}[theorem]{Definition}

\newtheorem{example}[theorem]{Example}

\newtheorem{observation}[theorem]{Observation}

\newcommand{\xhdr}[1]{\paragraph{\bf #1}}
\newcommand{\omt}[1]{}

\def\squarebox#1{\hbox to #1{\hfill\vbox to #1{\vfill}}}
\newcommand{\qed}{\hspace*{\fill}\vbox{\hrule\hbox{\vrule\squarebox{.667em}\vrule}\hrule}\smallskip}
\newenvironment{proof}{\noindent{\bf Proof:~~}}{\(\qed\)}

\newcommand{\io}{s}     
\newcommand{\no}{x}    
\newcommand{\oo}{y}    
\newcommand{\co}{c}    
\newcommand{\pc}{sc}    
\newcommand{\ao}{z}    

\begin{document}
\title{On Discrete Preferences and Coordination \footnote{
This work has been supported in part by
a Simons Investigator Award,
a Google Research Grant,
an ARO MURI grant,
and NSF grants
IIS-0910664, 
CCF-0910940, 
IIS-1016099 
and a Microsoft Research Fellowship.
}}
\author{  
Flavio Chierichetti
\thanks{
Sapienza University of Rome, Italy.
Email: flavio@di.uniroma1.it.
}
\and
Jon Kleinberg
\thanks{
Cornell University, Ithaca NY 14853.
Email: kleinber@cs.cornell.edu.
}
 \and 
Sigal Oren
\thanks{
Cornell University, Ithaca NY 14853.
Email: sigal@cs.cornell.edu.
}
}

\date{}

\begin{titlepage}
\maketitle
\begin{abstract}
An active line of research has considered games played on networks
in which payoffs depend on both a player's individual decision
and also the decisions of his or her neighbors.  
Such games have been used to model issues including the formation of opinions 
(in which people wish to express views consistent with those of 
their friends) and the adoption of technology (in which people or firms
seek compatibility with their network neighbors).

A basic question that has remained largely open in this area is to consider
games where the strategies available to the players come from a fixed,
discrete set, and where players may have 
different intrinsic preferences among the possible strategies.
It is natural to model the tension among these different preferences by
positing a distance function on the strategy set that determines a notion of
``similarity'' among strategies; a player's payoff is determined
by the distance from her chosen strategy to her preferred strategy
and to the strategies chosen by her network neighbors.
Even when there are only two strategies available,
this framework already leads to natural open questions about a
version of the classical Battle of the Sexes problem played on a graph;
such questions generalize issues in the study of network
coordination games.

We develop a set of techniques for analyzing this class of games,
which we refer to as {\em discrete preference games}.
We parametrize the games by the relative extent to which a player
takes into account the effect of her preferred strategy and the
effect of her neighbors' strategies, allowing us to interpolate
between network coordination games and unilateral decision-making. 
When these two effects are balanced,
we show that the price of stability is equal to 1 for any discrete preference
game in which the distance function on the strategies is a tree metric;
as a special case, this includes the Battle of the Sexes on a graph.
We also show that trees essentially form the maximal family of
metrics for which the price of stability is 1, and produce a 
collection of metrics on which the price of stability converges
to an asymptotically tight bound of 2.
\end{abstract}

\end{titlepage}


\section{Introduction}

People often make decisions in settings where the outcome 
depends not only on their own choices, but also on the choices of the people
they interact with.
A natural model for such situations is to consider a game played
on a graph that represents an underlying social network,
where the nodes are the players.
Each node's personal decision corresponds
to selecting a strategy, and the node's payoff depends on the
strategies chosen by itself and its neighbors in the graph
\cite{blume-statmech,ellison-coordination,morris-contagion}.

\xhdr{Coordination and Internal Preferences}
A fundamental class of such games involves payoffs 
based on the interplay between {\em coordination} --- each player has 
an incentive to match the strategies of his or her neighbors ---
and {\em internal preferences} --- each player also has an intrinsic preference
for certain strategies over others, independent of the desire to match
what others are doing.
Trade-offs of this type come up in a very broad collection of situations,
and it is worth mentioning several that motivate our work here.
\begin{itemize}
\item In the context of opinion formation, a group of people or
organizations might each possess different internal views, 
but they are willing to express or endorse a ``compromise'' opinion
so as to be in closer alignment with their network neighbors.
\item Questions involving technological compatibility among firms
tend to have this trade-off as a fundamental component:
firms seek to coordinate on shared standards
despite having internal cost structures that
favor different solutions.
\item Related to the previous example, a similar issue comes up 
in cooperative facility location problems, where firms have preferences
for where to locate, but each firm also wants to locate near the
firms with which it interacts.
\end{itemize}

A line of work beginning in the mathematical social sciences has
considered versions of this question --- often motivated by the 
first class of examples above, concerned with opinion formation ---
where the possible strategies
correspond to a continuous space such as $\Re^d$
\cite{friedkin-initial-opinions,jackson-networks-book}.
This makes it possible for players to adopt arbitrarily fine-grained
``average'' strategies from among any set of options, and
most of the dynamics and equilibrium properties of such models
are driven by this type of averaging.
In particular, dynamics based on repeated averaging have been shown 
in early work to exhibit nice convergence properties
\cite{friedkin-initial-opinions},
and recent work
including by two of the authors has developed bounds on the relationship
between equilibria and social optima
\cite{bindel-k-o-focs11}.

\xhdr{Discrete Preferences}
In many settings that exhibit a tension between coordination and
individual preferences, however, there is no natural way to
average among the available options.
Instead, the alternatives are drawn from a fixed discrete set ---
for example, there is only a given set of available technologies
for firms to choose among, or a fixed set of political candidates 
to endorse or vote for.
On a much longer time scale, there is always the possibility that
additional options could be created to interpolate between what's
available, but on the time scale over which the strategic interaction
takes place, the players must choose from among the discrete set of
alternatives that is available.

Among a small fixed set of players, 
coordination with discrete preferences is at the heart of a long line
of games in the economic theory literature --- perhaps
the most primitive example is the classic {\em Battle of the Sexes} game,
based on a pedagogical story 
in which one member of a couple wants to see movie $A$ while the other
wants to see movie $B$, but both want to go to a movie together.
This provides a very concrete illustration of a set of payoffs
in which the (two) players have (i) conflicting internal preferences
($A$ and $B$ respectively), (ii) an incentive to arrive at a compromise,
and (iii) no way to ``average'' between the available options.

But essentially nothing is known about the properties of the
games that arise when we consider such a payoff structure
in a network context.
Even the direct generalization of Battle of the Sexes (BoS) to a graph
is more or less unexplored in this sense ---
each node plays a copy of BoS on each of
its incident edges, choosing
a single strategy $A$ or $B$ for use in all copies, incurring a cost
from miscoordination with neighbors and an additional
fixed cost when the node's choice differs from its inherent preference.
Indeed, as some evidence of the complexity of even this formulation,
note that the version in which each node has an intrinsic preference
for $A$ is equivalent to the standard network coordination game,
which already exhibits rich graph-theoretic structure
\cite{morris-contagion}.
And beyond this, of course, lies the prospect of such games with
larger and more involved strategy sets.

\xhdr{Formalizing Discrete Preference Games}
In this paper, we develop a set of techniques for analyzing
this type of discrete preference games on a network,
and we establish tight bounds on the price of stability for
several important families of such games.

To formulate a general model for this type of game, we start with 
an undirected graph $G = (V,E)$ representing the network on the players,
and an underlying finite set $L$ of strategies.
Each player $i \in V$ has a {\em preferred strategy} $s_i \in L$,
which is what $i$ would choose in the absence of any other players.
Finally, there is a metric $d(\cdot,\cdot)$
on the strategy set $L$ --- that is, a distance function
satisfying (i) $d(i,i) = 0$ 
for all $i$, (ii) $d(i, j) = d(j, i)$ for all $i, j$,
and (iii) $d(i,j) \leq d(i,k)+d(k,j)$ for
all $i$, $j$ and $k$. 
For $i, j \in L$, the distance $d(i,j)$ 
intuitively measures how ``different'' $i$ and $j$ are as choices;
players want to avoid choosing strategies that are at large distance
from either their own internal preference or from the strategies chosen
by their neighbors.

Each player's objective is to minimize her cost (think of this
as the negative of her payoff): for a fixed parameter $\alpha \in [0,1]$,
the cost to player $i$ when players choose the strategy vector
$\ao = \langle \ao_j : j \in V \rangle$ is 
$$c_i(\ao) = \alpha \cdot d(\io_i,\ao_i) + 
\sum_{j \in N(i)} (1-\alpha) \cdot d(\ao_i,\ao_j),$$
where $N(i)$ is the set of neighbors of $i$ in $G$.
The parameter $\alpha$ essentially controls the extent to which 
players are more concerned with their preferred strategies or
their network neighbors; we will see that the behavior of the game
can undergo qualitative changes as we vary $\alpha$.

We say that the above formulation defines a 
{\em discrete preference game}.
Note that the network version of Battle of the Sexes described earlier
is essentially the special case in which $|L| = 2$, and 
network coordination games are the special case in which 
$|L| = 2$ and $\alpha = 0$, since then players are only concerned with 
matching their neighbors.
The case in which $d(\cdot,\cdot)$ is the distance metric among nodes on
a path is also interesting to focus on, since it
is the discrete analogue of the one-dimensional space of
real-valued opinions from continuous averaging models
\cite{bindel-k-o-focs11,friedkin-initial-opinions} --- 
consider for example the natural scenario in which a finite number of
discrete alternatives in an election are arranged along
a one-dimensional political spectrum.

We also note that discrete preference games belong to the 
well-known framework of graphical games, 
which essentially consist of games in which
the utility of every player depends only on the actions of its neighbors in
a network.
The interested reader is referred to the relevant chapter in
\cite{nisan-agt-book} and the references within.
In this context, Gottlob et al. proposed a generalization of
Battle of the Sexes (BoS) to a graphical setting \cite{Gottlob-hard-games},
but their formulation was much more complex than our starting point,
with their questions correspondingly
focused on existence and computational complexity,
rather than on the types of performance guarantees we will be seeking.

For any discrete preference game, 
we will see that it is possible to define an exact potential function,
and hence these games possess pure Nash equilibria.

\xhdr{Price of Stability in Discrete Preference Games}
We can also ask about the {\em social cost} of a strategy vector
$\ao = \langle \ao_j : j \in V \rangle$,
defined as the sum of all players' costs:
$$\co(\ao) =  \sum_{i \in V} \alpha \cdot d(\io_i,\ao_i) + 
2\sum_{(i,j) \in E} (1-\alpha) \cdot d(\ao_i,\ao_j).$$
We note that the problem of minimizing the social cost is
an instance of the {\em metric labeling problem}, in which we
want to assign labels to nodes in order to minimize a sum of
per-node costs and edge separation costs
\cite{boykov-fast-apx-min,kleinberg-met-label}.

Since an underlying motivation for studying this class of games
is the tension between preferred strategies and agreement on edges,
it is natural to study its consequences on the social cost via
the price of anarchy and/or the price of stability.
The price of anarchy is in fact too severe a measure
for this class of games; indeed, as we discuss in the next
section, it is already unbounded for the
well-studied class of network coordination games that our model contains
as a special case.  

We therefore consider the price of stability, which turns out
to impart a rich structure to the problem.
The price of stability is also natural in terms of the underlying
examples discussed earlier as motivation; in most of these
settings, it makes sense to propose a solution --- for example,
a compromise option in a political setting or a proposed set
of technology choices for a set of interacting firms --- and 
then to see if it is stable with respect to equilibrium.

\xhdr{Overview of Results}
As a starting point for reference, observe that network coordination
games (where players are not concerned with their preferred strategies)
clearly have a price of stability of $1$: the players can all
choose the same strategy and achieve a cost of $0$.
But even for a general discrete preference game with two strategies ---
i.e. Battle of the Sexes on a network --- the price of stability
is already more subtle, since the social optimum may have a more complex
structure (as a two-label metric labeling problem, and hence 
a minimum cut problem).

We begin by giving tight bounds on 
the maximum possible price of stability in the
two-strategy case as a function of the parameter $\alpha$.
The dependence on $\alpha$ has a complex non-monontonic character;
in particular, the price of stability is equal to $1$ for all instances
if and only if $\alpha \leq 1/2$ or $\alpha = 2/3$, and 
more generally the price of stability as a function of $\alpha$
displays a type of ``saw-tooth'' behavior with infinitely many
local minima in the interval $[0,1]$.
Our analysis uses a careful scheduling of the best-response dynamics
so as to track the updates of players toward a solution with low social cost.

Above we also mentioned the distance metric of a path
as a case of interest in opinion formation. 
We show that when 
$\alpha \leq 1/2$, the price
of stability for instances based on such metrics is always $1$,
by proving the stronger statement that in fact the price of stability
is always $1$ for any discrete preference game based on a tree metric.
Our analysis for tree metrics involves considering how
players' best responses lie at the medians of their neighbors' strategies
in the metric, and then developing combinatorial techniques for
reasoning about the arrangement of these collections of medians on 
the underlying tree.

Like path metrics, tree metrics are also relevant to motivating scenarios
in terms of opinion formation, when individuals classify the space
of possible opinions according to a hierarchical structure rather
than a linear one. 
To take one example of this, 
consider students choosing a major in college, 
where each student has an internal preference and an interest in 
picking a major that is similar to the choices of her friends. 
The different subjects roughly follow a hierarchy --- on top we might have science, engineering, and humanities; under science we can have for 
example biology, physics, and other areas; and under biology 
we can have subjects including genetics and plant breeding. 
This setting fits our model since each person has some internal inclination for a major, but still it is arguably the case that a math major has more 
in common in her educational experience with her computer science friends 
than with her friends in comparative literature.

The two families of instances described above (two strategies and
tree metrics) both have price of stability equal to $1$ when $\alpha \leq 1/2$.
But the price of stability can be greater than $1$ for more general
metrics when $\alpha \leq 1/2$.  
It is not hard to show (as we do in the next section)
that the price of stability is always at most $2$ for all $\alpha$,
and we match this bound by constructing and 
analyzing examples, based on perturbations of
uniform metrics, showing that the price of stability can 
be arbitrarily close to $2$ when $\alpha = 1/2$.

Finally, we consider a generalization of our model of discrete
preference games, which we term an {\em anchored preference game}.
Suppose that nodes are partitioned into two types: there are 
{\em fixed nodes} $i$ that have a preferred strategy anchored on
a particular value $\io_i$, and there are {\em strategic nodes}
that have no preferred strategy, so the cost of such a node $i$ is purely the
term $\sum_{j \in N(i)} d(\ao_i,\ao_j)$.
Only the strategic nodes choose strategies, and only they take
part in the definition of the equilibrium.

Anchored preference games are also of interest in their own right,
and not just as a generalization  of discrete
preference games,
since they can model settings where certain parts of the network
represent unmodifiable constraints --- for example, how
technological compatibility might have to take into account 
certain unchangeable factors, or how individuals adapting their
opinions to each other might also be taking into account 
agents such as media sources that are in effect outside the
immediate strategic environment.
We generalize our results for tree metrics to the case of
anchored preference games, parametrizing the price of stability
by the maximum number of fixed nodes among the neighbors of
any strategic node.

\section{Preliminaries}

Recall that in a discrete preference
game played on a graph $G = (V,E)$ with strategy set $L$,
each player $i \in V$ has a preferred strategy $\io_i \in L$.
The cost incurred by player $i$ when all players choose strategies
$\ao = \langle \ao_j : j \in V \rangle$ is
$$c_i(\ao) = \alpha \cdot d(\io_i,\ao_i) +
\sum_{j \in N(i)} (1-\alpha) \cdot d(\ao_i,\ao_j).$$
The social cost of the game is the sum of all the players' costs:
$$\co(\ao) =  \sum_{i \in V} \alpha \cdot d(\io_i,\ao_i) + 
2\sum_{(i,j) \in E} (1-\alpha) \cdot d(\ao_i,\ao_j).$$
Another quantity that is useful to define is the contribution of player $i$ to the social cost -- by this we quantify not only the cost player $i$ is exhibiting but also the cost it is inflicting on its neighbors: 
$$\pc_i(\ao)  = \alpha \cdot d(\io_i,\ao_i) + 2\sum_{j \in N(i)} (1-\alpha) \cdot d(\ao_i,\ao_j).$$
As is standard, we denote by $\ao_{-i}$ the strategy vector $\ao$ without the $i^{th}$ coordinate.

We first show that this class of games includes
instances for which the price of anarchy (PoA) is unbounded. 
A simple instance for which the PoA is unbounded is one in which the preferred strategy of all the players is the same -- thus the cost of the optimal solution is $0$, and it is also an equilibrium for all the players to play some other strategy and incur a positive cost.\footnote{This type of equilibrium,
in addition to simply producing an unbounded PoA, has a natural
interpretation in our motivating contexts.  In technology adoption,
it corresponds to convergence on a standard that no firm individually wants,
but which is hard to move away from once it has become the consensus.
In opinion formation, it corresponds essentially to a kind of
``superstitious'' belief that is universally expressed, and hence is
hard for people to outwardly disavow even though they prefer an alternate
opinion.}
In the next claim, we use this idea to construct, for every value 
of $\alpha<1$, an instance for which the PoA is unbounded.

\begin{claim}
For any $\alpha < 1$ there exists an instance for which the price of anarchy is unbounded.
\end{claim}
\begin{proof}
Assume the strategy space contains two strategies $A$ and $B$, such that $d(A,B)=1$. For any $0 < \alpha < 1$ we consider a clique of size $\lceil \frac{\alpha}{1-\alpha} \rceil +1 $ in which all players' preferred strategy is $A$ and show it is an equilibrium for all the players to play strategy $B$. To see why, observe that if the rest of the players play strategy $B$, then player $i$'s cost for playing strategy $A$ is $(1-\alpha) \cdot \lceil \frac{\alpha}{1-\alpha} \rceil$ which is at least $\alpha$. Since the cost of player $i$ for playing strategy $B$ is $\alpha$ we have that it is an equilibrium for all players to play strategy $B$. The PoA of such an instance is unbounded as the cost of the equilibrium in which all players play strategy $B$ is strictly positive but the cost of the optimal solution is $0$.

To show that the PoA can be unbounded for $\alpha=0$, a slightly different instance is required, which will be familiar from the literature of network coordination games.
When players do not have a preference the optimal solution is clearly for all players to play the same strategy, as such a solution has a cost of $0$. However, Figure \ref{fig:poa-example} depicts an instance for which there exists a Nash equilibrium in which not all the players play the same strategy and hence the cost of this equilibrium is strictly positive.
\end{proof}

\begin{figure}[htb]
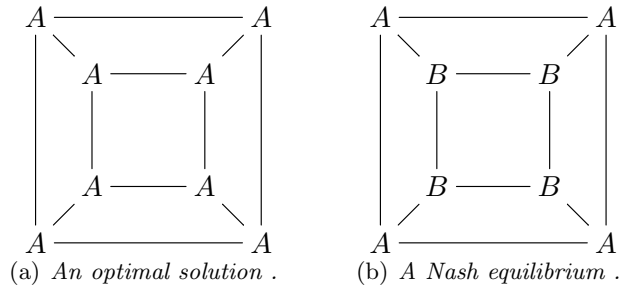

\begin{center}
\subfigure[\emph{An optimal solution .}]{
\xygraph{ !{<0cm,0cm>;<1.5cm,0cm>:<0cm,1.5cm>::} !{(0,0) }*+{A}="1" !{(2,0) }*+{A}="2" !{(2,2) }*+{A}="3" !{(0,2) }*+{A}="4"
!{(0.5,1.5) }*+{A}="5" !{(1.5,1.5) }*+{A}="6" !{(1.5,0.5) }*+{A}="7"  !{(0.5,0.5) }*+{A}="8" "1"-"2" "2"-"3" "3"-"4" "4"-"1" "5"-"6" "6"-"7" "7"-"8" "8"-"5" "4"-"5" "3"-"6"
"7"-"2" "1"-"8" } 
} \hspace{7mm}
\subfigure[\emph{A Nash equilibrium .}]{
\xygraph{ !{<0cm,0cm>;<1.5cm,0cm>:<0cm,1.5cm>::} !{(0,0) }*+{A}="1" !{(2,0) }*+{A}="2" !{(2,2) }*+{A}="3" !{(0,2) }*+{A}="4"
!{(0.5,1.5) }*+{B}="5" !{(1.5,1.5) }*+{B}="6" !{(1.5,0.5) }*+{B}="7"  !{(0.5,0.5) }*+{B}="8" "1"-"2" "2"-"3" "3"-"4" "4"-"1" "5"-"6" "6"-"7" "7"-"8" "8"-"5" "4"-"5" "3"-"6"
"7"-"2" "1"-"8" } 
} 
\caption{
{\small
An instance illustrating that the PoA can be unbounded even when the players do not have a preferred strategy (i.e., $\alpha = 0$).
\label{fig:poa-example}
}
}
\end{center}
\vspace*{-0.2in}
\end{figure}

We note that the worst equilibrium in the previous instances is not a strong equilibrium. Thus, if all the players could coordinate a joint deviation to strategy $A$ they can all benefit. 
A natural question is what happens if we restrict ourselves to 
the worst strong Nash equilibrium (strong PoA), in which a
simultaneous deviation by a set of players is allowed.
Unfortunately, the strong PoA can still be quite high (linear in the number of players). Take for example $\alpha < \frac 1 2$ and consider a clique of size $n$ in which all but one of the players prefer strategy $A$. 
In this case it is not hard to verify that the equilibrium in which all players play strategy $B$ is strong. 
The reason is that the player that prefers strategy $B$ cannot gain from any deviation, and the cost of any other player would increase 
by at least $1-\alpha - \alpha > 0$ for any deviation. 
The cost of such an equilibrium is $\alpha (n-1)$ in comparison to the optimal solution which has a cost of $\alpha$.  

As we just showed both the PoA and the strong PoA can be very high,
and hence for the remainder of the paper
we focus on the qualities of the best Nash equilibrium, 
trying to bound the price of stability (PoS). 
We begin by showing that the price of stability is bounded by $2$. This is done by a potential function argument which also proves that a Nash equilibrium always exists.
\begin{claim}
The price of stability is bounded by $2$.
\label{claim:pos-2}
\end{claim}
\begin{proof}
We first prove that the following function is an exact potential function:
$$ \phi(\ao) = \alpha \sum_{i \in V} d(\ao_i,\io_i) + (1-\alpha) \sum_{(i,j) \in E} d(\ao_i,\ao_j). $$
To see why, note that: $\phi(\ao_i,\ao_{-i}) - \phi(\ao'_i,\ao_{-i})=$
\begin{align*}
&\alpha \cdot d(\ao_i,\io_i) + (1-\alpha) \sum_{j \in N(i)} d(\ao_i,\ao_j) -  \Big(\alpha \cdot d(\ao'_i,\io_i) + (1-\alpha) \sum_{j \in N(i)} d(\ao'_i,\ao_j) \Big) \\
 &= \co_i(\ao_i,\ao_{-i}) - \co_i(\ao'_i,\ao_{-i}).
\end{align*}
Denote by $\no$ the global minimizer of the potential function and by $\oo$ the optimal solution. By definition $\no$ is an equilibrium and it provides a $2$-approximation to the optimal social cost since $\co(\no) \leq 2 \phi(\no) \leq 2 \phi(\oo) \leq 2\co(\oo)$.
\end{proof}


\section{The Case of Two Strategies: Battle of the Sexes on a Network}
We begin by considering the subclass of instances in which the players only have two different strategies $A$ and $B$. Without loss of generality we assume that $d(A,B)=1$. We denote by $N_j(i)$ the set of $i$'s neighbors using strategy $j$ and by $\bar \io_i$ the strategy opposite to $\io_i$. When the strategy space contains only two strategies, a player's best response is to pick a strategy which is the weighted majority of its own preferred strategy and the strategies played by its neighbors. The next two observations formalize this statement and a similar statement regarding a player's strategy minimizing the social cost:
\begin{observation} \label{obs:best_resp}
The strategy $\io_i$ minimizes player $i$'s cost ($\co_i(\ao)$) if:
\begin{align*} 
(1-\alpha)N_{\bar \io_i}(i) \leq \alpha +(1-\alpha)N_{\io_i}(i) \implies N_{\bar \io_i}(i) \leq \dfrac{\alpha}{1-\alpha} + N_{\io_i}(i).
\end{align*}
\end{observation}
\begin{observation} \label{obs:best_cost}
The strategy $\io_i$ minimizes the social cost ($\pc_i(\ao)$) if:
\begin{align*}
2(1-\alpha)N_{\bar \io_i}(i) \leq \alpha +2(1-\alpha)N_{\io_i}(i) \implies N_{\bar \io_i}(i) \leq \dfrac{\alpha}{2(1-\alpha)} + N_{\io_i}(i).
\end{align*}
\end{observation}

We present a simple best response order that results in a Nash equilibrium after a linear number of best responses. We will later see how this order can be used to bound the PoS.
\begin{lemma} \label{lem:best_dyn}
Starting from some initial strategy vector, the following best response order results in a Nash equilibrium:
\begin{enumerate}
\item While there exists a player that can reduce its cost by changing its strategy to $A$, let it do a best response. If there is no such player continue to the second step.
\item While there exists a player that can reduce its cost by changing its strategy to $B$, let it do a best response.
\end{enumerate}
\end{lemma}
\begin{proof}
To see why the resulting strategy vector is a Nash equilibrium, observe that after the first step, all nodes are either satisfied with their current strategy choice, or can benefit from changing their strategy to $B$. This property remains true after some of the nodes change their strategy to $B$ since the fact that a node has more neighbors using strategy $B$ can only reduce the attractiveness of switching to strategy $A$. Thus, at the end of the second step all nodes are satisfied.
\end{proof}

Next, we characterize the values of $\alpha$ for which the price of stability is $1$.

\begin{claim} \label{clm:opt}
If $\alpha \leq \frac 1 2$ or $\alpha=\frac 23$, then in any instance there exists an optimal solution which is also a Nash equilibrium.
\end{claim}
\begin{proof}
Let $\oo$ be an optimal solution minimizing the potential function $\phi(\cdot)$. Assume towards a contradiction that it is not a Nash equilibrium. Let player $i$ be a player that prefers to switch to a different strategy than $\oo_i$. Denote player $i$'s best response by $\no_i$. By Observations \ref{obs:best_resp} and \ref{obs:best_cost}, it is easy to see that if $\oo_i = \io_i$, 
then the strategy minimizing player $i$'s cost is also $\io_i$. Thus, we have that $\oo_i \neq \io_i$ and $\no_i=\io_i$. If $\io_i$ is a minimizer of the social cost function then $(\io_i,\oo_{-i})$ is also an optimal solution. 
This contradicts the assumption that $\oo$ is a minimizer of the potential function $\phi(\cdot)$ since $\phi(\io_i,\oo_{-i}) < \phi(\oo)$. Therefore by Observations \ref{obs:best_resp} and \ref{obs:best_cost} we have that $N_{\bar \io_i}(i) < \dfrac{\alpha}{1-\alpha} + N_{\io_i}(i)$ and $\dfrac{\alpha}{2(1-\alpha)} + N_{\io_i}(i) < N_{\bar \io_i}(i)$. By combining the two inequalities we get that:
\begin{align*}
\dfrac{\alpha}{2(1-\alpha)} + N_{\io_i}(i) < N_{\bar \io_i}(i) <  \dfrac{\alpha}{1-\alpha} + N_{\io_i}(i).
\end{align*}

Since $N_{\bar \io_i}(i)$ and $N_{\io_i}(i)$ are both integers, this implies that there exists an integer $k$ such that $\dfrac{\alpha}{2(1-\alpha)}  < k <  \dfrac{\alpha}{1-\alpha}$. This holds for $k=1$ if $\dfrac{\alpha}{2(1-\alpha)}  < 1 <\dfrac{\alpha}{1-\alpha} $ (implying $\frac 12 < \alpha < \frac 23$) or for some $k>1$ if $\dfrac{\alpha}{1-\alpha} -\dfrac{\alpha}{2(1-\alpha)} >1 $ (implying $ \alpha > \frac 23$). Thus, for $\alpha \leq \frac 12$ or $\alpha=\frac 23$ any instance admits an optimal solution which is also a Nash equilibrium.
\end{proof}

It is natural to ask what is the PoS for the values of $\alpha$ for which we know the optimal solution is not a Nash equilibrium. The following theorem provides an answer to this question by computing the ratio between the optimal solution and a Nash equilibrium obtained by performing the sequence of best responses Lemma \ref{lem:best_dyn} prescribes.
\begin{theorem} \label{thm:2op:pos}
For $\frac 12 <\alpha < 1$, $PoS \leq 2 \Big\lceil \dfrac{\alpha}{1-\alpha} -1 \Big\rceil \cdot \dfrac{1-\alpha}{\alpha}$.
\end{theorem}
\begin{proof}
Let $\no$ be the equilibrium achieved by the sequence described in Lemma \ref{lem:best_dyn} starting from an optimal solution $\oo$. Denote the strategy vector at the end of the first step by $\no_1$ and at the end of the sequence by $\no_2 = \no$. We assume that a player performs a best response only when it can strictly decrease the cost by doing so, thus we only reason about the case where the player's best response is unique. In the following Lemma we bound the increase in the social cost inflicted by the players' unique best responses (the proof can be found below):
\begin{lemma} \label{lem:best_response}
Let player $i$'s \emph{unique} best response when the rest of the players play $\ao_{-i}$ be $\no_i$ then:
\begin{enumerate}
\item If $\no_i =\bar{\io_i}$ then $c(\bar{\io_i},\ao_{-i}) - c(\io_i, \ao_{-i}) \leq \alpha -2(1-\alpha)\Big\lfloor \dfrac{\alpha}{1-\alpha}+1\Big\rfloor $. \label{enum:anti_pref}
\item If $\no_i =\io_i$ then $c(\io_i,\ao_{-i}) - c(\bar{\io_i}, \ao_{-i}) \leq -\alpha+ 2(1-\alpha)\Big\lceil \dfrac{\alpha}{1-\alpha}-1 \Big\rceil $.  \label{enum:pref}
\end{enumerate}
\end{lemma}
Notice that by statement (\ref {enum:anti_pref}) of Lemma \ref{lem:best_response} a node that changes its strategy to a strategy different than its preferred strategy can only reduce the social cost. Also, note that if a node changes its strategy in the first step to $A$ and in the second step back to $B$ its total contribution to the social cost is non-positive. The reason is that in one of these changes the player changed its strategy from $\io_i$ to $\bar{\io_i}$ and in the other from $\bar{\io_i}$ to $\io_i$. The effect of these two changes on the social cost sums up to $2(1-\alpha)\left( \lceil \dfrac{\alpha}{1-\alpha} \rceil -\lfloor \dfrac{\alpha}{1-\alpha} \rfloor -2 \right) \leq 0$. Thus we can ignore such changes as well. 

The only nodes that are capable of increasing the social cost by performing a best response are ones that play in the optimal solution a different strategy than their preferred strategy ($\oo_i \neq \io_i$). By definition their number equals exactly $\sum_i d(\oo_i,\io_i)$ as $d(\oo_i,\io_i) = 1$ if $\oo_i \neq \io_i$ and $0$ otherwise. Statement (\ref{enum:pref}) of Lemma \ref{lem:best_response} guarantees us that each of these nodes can increase the social cost by at most $-\alpha+ 2(1-\alpha)\Big\lceil \dfrac{\alpha}{1-\alpha}-1 \Big\rceil $. Thus, we get the following bound: 
\begin{align*} 
c(\no) &\leq \co(\oo) + \left(-\alpha+ 2(1-\alpha)\Big\lceil \dfrac{\alpha}{1-\alpha}-1 \Big\rceil \right) \sum_{i \in V} d(\oo_i,\io_i) \\
&= 2(1-\alpha) \sum_{(i,j)\in E} d(\oo_i,\oo_j)  + 2 (1-\alpha) \Big\lceil \dfrac{\alpha}{1-\alpha} -1 \Big\rceil \sum_{i \in V} d(\oo_i,\io_i).
\end{align*}

We are now ready to compute the bound on the PoS:
\begin{align*}
PoS &\leq \dfrac{2\Big\lceil \dfrac{\alpha}{1-\alpha} -1 \Big\rceil \cdot (1-\alpha)  \sum_{i \in V} d(\oo_i,\io_i) + 2(1-\alpha) \sum_{(i,j)\in E} d(\oo_i,\oo_j) }{\alpha \sum_{i \in V} d(\oo_i,\io_i)  +2(1-\alpha) \sum_{(i,j)\in E} d(\oo_i,\oo_j)} \\
&\leq \dfrac{2\Big\lceil \dfrac{\alpha}{1-\alpha} -1 \Big\rceil \cdot \dfrac{ 1-\alpha}{\alpha} \cdot \left( \alpha  \sum_{i \in V} d(\oo_i,\io_i) + 2(1-\alpha) \sum_{(i,j)\in E} d(\oo_i,\oo_j) \right) }{\alpha \sum_{i \in V} d(\oo_i,\io_i)  +2(1-\alpha) \sum_{(i,j)\in E} d(\oo_i,\oo_j)} \\
&\leq 2\Big\lceil \dfrac{\alpha}{1-\alpha} -1 \Big\rceil \cdot \dfrac{ 1-\alpha}{\alpha}.
\end{align*}
\end{proof}

\noindent \textbf{Proof of Lemma \ref{lem:best_response}: }
Notice we are only considering the effect player $i$'s strategy has on the social cost, thus we have that: $c(\bar{\io_i},\ao_{-i}) - c(\io_i, \ao_{-i}) = \pc_i(\bar{\io_i},\ao_{-i}) - \pc_i(\io_i, \ao_{-i})$. Implying that 
$c(\bar{\io_i},\ao_{-i}) - c(\io_i, \ao_{-i}) =  \alpha + 2(1-\alpha)( N_{\io_i}(i) -N_{\bar \io_i}(i))$.
For proving statement (\ref{enum:anti_pref}) we observe that since $\bar \io_i$ is player $i$'s unique best response then: $N_{\bar \io_i}(i) > \dfrac{\alpha}{1-\alpha} + N_{\io_i}(i)$ as player $i$ has a strictly smaller cost for playing strategy $\bar{\io_i}$ than for playing strategy $\io_i$. Since, $N_{\bar \io_i}(i)$ and $N_{\io_i}(i)$ are integers this implies that $N_{\bar \io_i}(i) -N_{\io_i}(i)  \geq \lfloor \dfrac{\alpha}{1-\alpha}+1\rfloor$ and the bound is achieved.

For proving statement (\ref{enum:pref}), observe that $c(\io_i,\ao_{-i}) - c(\bar{\io_i}, \ao_{-i}) = -\alpha + 2(1-\alpha)( N_{\bar \io_i}(i) - N_{\io_i}(i))$. Now since $\io_i$ is player $i$'s best response we have that: $N_{\bar \io_i}(i) < \dfrac{\alpha}{1-\alpha} + N_{\io_i}(i)$ as player $i$'s best response is to use strategy $\io_i$. Which similarly to the previous bound implies that 
$N_{\bar \io_i}(i) -N_{\io_i}(i)  \leq \lceil \dfrac{\alpha}{1-\alpha}-1 \rceil$ as required.
$\qed$

It is interesting to take a closer look at the upper bound on the PoS we computed (as we will see later this bound is tight). In Figure \ref{fig:pos} we plot the upper bound on the PoS as a function of $\alpha$. We can see that as $\alpha$ approaches $1$ the PoS approaches $2$ and also that for any $k\geq 2$, as $\epsilon$ approaches $0$, the PoS of $\alpha = \dfrac{k-1}{k} + \epsilon$ also approaches to $2$. This uncharacteristic saw-like behavior of the PoS originates from the fact that for every value of $\alpha$ the maximal PoS is achieved by a star graph. This is proved in the following claim. 

\begin{figure}[htb]
\begin{center}
\includegraphics[width=2.5in]{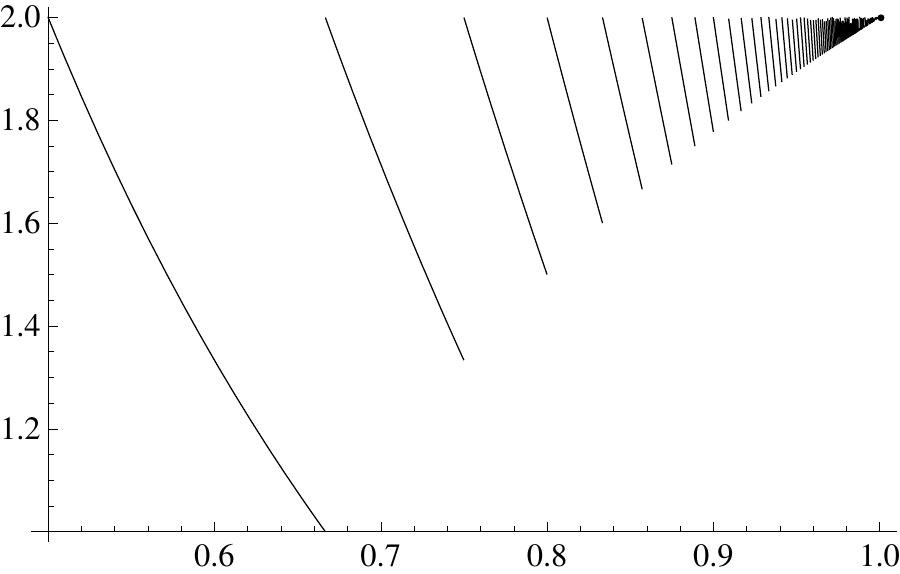}
\end{center}
\caption{The tight upper bound on the PoS for two strategies as a function of $\alpha$ for the range $\frac 12 < \alpha < 1$.}
\label{fig:pos}
\end{figure}

\begin{claim} \label{clm:star}
For any $\alpha>1/2$, $\alpha \neq 2/3$ there exists an instance achieving a price of stability of $ 2 \Big\lceil \dfrac{\alpha}{1-\alpha} -1 \Big\rceil \cdot \dfrac{1-\alpha}{\alpha}$.
\end{claim}
\begin{proof}
Consider a star consisting of $\Big\lceil \dfrac{\alpha}{1-\alpha} -1 \Big\rceil$ peripheral nodes that prefer strategy $A$ and a central node that prefers strategy $B$. In the optimal solution the central node plays strategy $A$ for a cost of $\alpha$. However, this is not a Nash equilibrium since for playing strategy $B$ it exhibits a cost of $(1-\alpha) \cdot \Big\lceil \dfrac{\alpha}{1-\alpha} -1 \Big\rceil < \alpha$. Thus, the central node prefers to play its preferred strategy.
\end{proof}
\begin{corollary}
As $n$ goes to infinity the PoS for $\alpha = \dfrac{n-1}{n}$ approaches $2$. 
\end{corollary}


\newcommand{\sbr}{SC}    
\newcommand{\br}{C}    

\section{ Richer Strategy Spaces }

In the previous section we have seen that even when there are only two strategies in the game (the Battle of the Sexes on a network), for at least some values of $\alpha>\frac 12$, 
the PoS can be quite close to $2$.

These bounds carry over to larger strategy spaces since 
an instance can always use only two strategies from the strategy space. 
However, for $\alpha \leq \frac 12$ the PoS for the Battle of the Sexes on a
network is $1$, so a natural question is how bad the PoS can
be once we have more strategies in the space. This is the question we deal with for the rest of the paper. 

\subsection{Tree Metrics}

We begin by considering the case in which the distance function on the strategy
set is a tree metric, defined as the shortest-path
metric among the nodes in a tree.  
(As such, tree metrics are a special 
case of graphic metrics, in which there is a graph on the elements of the space and the distance between every two elements is defined to be the length of the shortest path between them in the graph.)
We show that if the distance function is a tree metric then the price of stability is $1$ for any rational $\alpha \leq \frac 12$. 

Denote by $\br_i(\ao)$ and $\sbr_i(\ao) $ the strategies of player $i$ that minimize $c_i(\ao) = \alpha \cdot d(\ao_i,\io_i) + (1-\alpha)\sum_{j \in N(i)} d(\ao_i,\ao_j)$ and $\pc_i(\ao) = \alpha \cdot d(\ao_i,\io_i) + 2(1-\alpha)\sum_{j \in N(i)} d(\ao_i,\ao_j)$ respectively. We show that if for every player $i$ the intersection of the two sets  $\br_i(\ao)$ and $\sbr_i(\ao) $ is always non-empty then the price of stability is $1$:

\begin{claim} \label{clm:best_response}
If for every player $i$ and strategy vector $\ao$, $\sbr_i(\ao) \cap \br_i(\ao) \neq \emptyset$, then $PoS=1$.
\end{claim}
\begin{proof}
First, recall the potential function $\phi(\cdot)$ used in
the proof of Claim \ref{claim:pos-2}, and consider
an optimal solution $\oo$ minimizing this potential function $\phi(\cdot)$. 
If $y$ is also a Nash equilibrium then we are done. Else, there exists a node $i$ that can strictly reduce its cost by performing a best response. By our assumption node $i$ can do this by choosing a strategy $\no_i \in  \sbr_i(\oo) \cap \br_i(\oo) $. The fact that $\no_i \in \sbr_i(\oo)$ implies that the change in strategy of player $i$ does not affect the social cost. Therefore, $(\no_i,\oo_{-i})$ is also an optimal solution and $\phi(\oo)> \phi(\no_i,\oo_{-i})$, in contradiction to the assumption that $\oo$ is the optimal solution minimizing $\phi(\cdot)$.
\end{proof}

Our goal now is to show that the conditions of Claim \ref{clm:best_response} hold for a tree metric. Our first step is relating the strategies that consist of a player's best response (or social cost minimizer) and the set of medians of a node-weighted tree.
\begin{definition} [medians of a tree] 
Given a tree $T$ where the weight of node $v$ is denoted
$w(v)$, the set of $T$'s medians is $M(T) = \arg\min_{u \in V} \{ \sum_{v \in V} w(v) \cdot d(u,v) \} $.
\end{definition}

\begin{definition}
Given a network $G$, a tree metric $T$, a strategy vector $\ao$, a player $i$ and non-negative integers $q$ and $r$, we denote by $T_{i,\ao}(q,r)$ the tree $T$ with the following node weights:
$$
w(v) = \begin{cases}

  q + r\cdot |\{j\in N(i) | \ao_j = v \}| & \text{for $v=\io_i$} \\

  r\cdot |\{j\in N(i) | \ao_j = v \}| & \text{for $v \neq \io_i$} \\

\end{cases}
$$
\end{definition}

Next we show that for $\alpha = \frac{a}{a+b}$, every player $i$ and strategy vector $\ao$, it holds that $M(T_{i,\ao}(a,b)) = \br_i(\ao)$. To see why, observe that by construction we have $M(T_{i,\ao}(a,b)) =$
\begin{align*}
&\arg\min_{u \in V} \{ (a + b\cdot |\{j\in N(i) | \ao_j = \io_i \}|) \cdot d(u,\io_i) +   \sum_{v\neq \io_i \in V} b \cdot |\{j\in N(i) | \ao_j = v \}| \cdot d(u,v) \} \\
&=  \arg\min_{u \in V} \{ a \cdot d(u,\io_i) +   b \sum_{j \in N(i)} d(u,\ao_j) \} = \br_i(\ao).
\end{align*}
Similarly, it is easy to show that $M(T_{i,\ao}(a,2b)) = \sbr_i(\ao)$. 
Thus, to show that $\sbr_i(\ao) \cap \br_i(\ao) \neq \emptyset$ it is sufficient to show that $T_{i,\ao}(a,b)$ and $T_{i,\ao}(a,2b)$ share a median. This is done by using the following proposition:

\begin{proposition} \label{prop:median_facts}
Let $T_1$ and $T_2$ be two node-weighted trees with the same edges and nodes, then:
\begin{itemize}
\item If there exists a node $v$, such that for every $u\neq v \in V$, we have $w_1(u)=w_2(u)$ and for $v$ we have $|w_1(v)-w_2(v)| = 1$, then $T_1$ and $T_2$ share a median.
\item If $T_1$ and $T_2$ share a median, then it is also a median of their union $T_1 \cup T_2$. Where the union of $T_1$ and $T_2$ is a tree with the same nodes and edges where the weight of node $v$ is $w_{1+2}(v) = w_1(v)+w_2(v)$.
\end{itemize}
\end{proposition}
The proof of this proposition is based on combinatorial claims from Section \ref{sub:sep} showing that a tree's medians and separators (defined below) coincide and establishing connections between the separators of different trees. Given Proposition \ref{prop:median_facts} we can now show that $T_{i,\ao}(a,b)$ and $T_{i,\ao}(a,2b)$ share a median:
\begin{lemma}
For $\alpha =\frac{a}{a+b} \leq \frac 12$, every player $i$ and strategy vector $\ao$, $M(T_{i,\ao}(a,b)) \cap M(T_{i,\ao}(a,2b)) \neq \emptyset$.
\label{lem:intersect}
\end{lemma}
\begin{proof}
First, observe that by Proposition \ref{prop:median_facts} we have that $T_{i,\ao}(0,1)$ and $T_{i,\ao}(1,1)$ share a median. As medians are invariant to scaling this implies that $T_{i,\ao}(0,b-a)$ and $T_{i,\ao}(a,a)$ also share a median. Next, by the second statement of Proposition \ref{prop:median_facts} we have that any median they share is also a median of their union $T_{i,\ao}(a,b)$; let us 
denote this median by $u$. 
Since $u$ is a median of $T_{i,\ao}(0,b)$ and $T_{i,\ao}(a,b)$,
it is also a median of $T_i(a,2b)$ by applying Proposition \ref{prop:median_facts} again. 
Thus we have that $u$ is a median of both $T_{i,\ao}(a,b)$ and $T_{i,\ao}(a,2b)$ which concludes the proof.
\end{proof}

Hence we have proven the following theorem:

\begin{theorem}
If the distance metric is a tree metric then for rational  $\alpha \leq \frac 1 2$, there exists an optimal solution which is also a Nash equilibrium (PoS=1).
\label{thm:tree_metric}
\end{theorem}

\subsection{Combinatorial Properties of Medians and Separators in Trees} \label{sub:sep} 

We now state and prove the combinatorial facts about medians and
separators in trees that we used for the analysis above.
This builds on the highly tractable structure of
medians in trees developed in early work; see
\cite{gavish-tree-median,goldman-tree-median} and
the references therein.

Consider a tree where all nodes have integer weights, and
denote the weight of the tree by $w(V) = \sum_{v \in V} w(v)$. 
\begin{definition}
A separator of a tree $T$ is a node $v$ such that the weight of each connected component of $T-v$ is at most $w(V)/2$.
\end{definition}

\begin{claim} \label{clm:sep}
A node $u$ is a median of a tree $T$ if and only if it is a separator of $T$.
\end{claim}
\begin{proof}
Let $u$ be a median of a tree $T$, and assume towards a contradiction that it is not a separator; that is, there exists a component of $T-u$ of weight strictly greater than $w(V)/2$. 
Let $v$ be the neighbor of $u$ in this component. Consider locating the median at $v$. This reduces the distance to a total node weight of at least $w(V)/2 + 1/2$ by $1$, and increases the distance to less than a total node weight of $w(V)/2$ by $1$. Hence the sum of all distances decreases, and this 
contradicts the fact that $u$ is a median. Thus, every median of a tree is also a separator.

To show that any separator is also a median, we show that for any two separators $u_1$ and $u_2$ it holds that $\sum_{v \in V} w(v) \cdot d(u_1,v) = \sum_{v \in V} w(v) \cdot d(u_2,v)$. Since, we know that there exists a median which is a separator this will imply that any separator is a median. 

Denote by $C$ the connected component of the graph $T-u_1-u_2$ that includes the nodes on the path between $u_1$ and $u_2$. If $u_1$ and $u_2$ are adjacent let $C=\emptyset$. Denote the connected component of $T -u_1-C$ that includes $u_2$ by $C_2$ and the connected component of $T -u_2-C$ that includes $u_1$ by $C_1$. Note that by construction $C,C_1,C_2$ are disjoint and the union of their nodes equals $V$. Since $u_1$ is a separator it holds that $w(C_2)+w(C) \leq  w(V)/2$. This in turn implies that $w(C_1) \geq w(V)/2$. Similarly, since $u_2$  is a separator it holds that $w(C_1)+w(C) \leq  w(V)/2$. This in turn implies that $w(C_2) \geq w(V)/2$. Therefore, it has to be the case that $w(C)=0$, $w(C_1) = w(C_2) = w(V)/2$. We next show this implies that $\sum_{v \in V} w(v) \cdot d(u_1,v) = \sum_{v \in V} w(v) \cdot d(u_2,v)$. Observe that:
\begin{align*}
\sum_{v \in V} w(v) \cdot d(u_1,v) &= \sum_{v \in C_1} w(v) \cdot d(u_1,v) +  \sum_{v \in C_2 } w(v) \cdot \big( d(u_1,u_2)+ d(u_2,v) \big) \\
& =\sum_{v \in C_1} w(v) \cdot d(u_1,v) +  \sum_{v \in C_2} w(v) \cdot d(u_2,v) + w(C_2) \cdot d(u_1,u_2). 
\end{align*}
and similarly that:
\begin{align*}
\sum_{v \in V} w(v) \cdot d(u_2,v) &= \sum_{v \in C_2} w(v) \cdot d(u_2,v) +  \sum_{v \in C_1} w(v) \cdot d(u_1,v) + w(C_1) \cdot d(u_1,u_2). 
\end{align*}
The claim follows as we have shown that $w(C_2) = w(C_1)$.
\end{proof}

Next, we prove two claims relating the separators of different trees. We first show that if two trees differ only in the weight of a single node and the difference in weight of this node in the two trees is $1$ -- then they share a separator:
\begin{claim} \label{clm:diff_by_one}
Consider two trees $T_1$ and $T_2$ with the same set of edges and nodes that differ only in the weight of a single node $v$ - such that $w_2(v) = w_1(v)+1$. Then, $T_1$ and $T_2$ share a separator. 
\end{claim}
\begin{proof}
We first handle the case where $w_1(V)$ is odd. Let $u$ be a separator of $T_1$, then, in this case the size of each component in $T_1-u$ is at most $w_1(V)/2-1/2$. Thus for the same separator $u$ in $T_2$ the size of each component is at most $w_1(V)+1/2 = w_2(V)/2$. Therefore $u$ is still a separator. Assume that $w_1(V)$ is even. This implies that $w_1(V)+1$ is odd. Consider a separator $u'$ of $T_2$. Then the size of each connected component in $T_2-u'$ is at most $w_2(V)/2+1/2$, since $w_1(V)$ is even, this implies that the weight of each connected component is bounded by $w_1(V)/2$ and therefore $u'$ is also a separator of $T_1$.
\end{proof}

We show that if $u$ is a separator of both $T_1$ and $T_2$ it is also the separator of their union:
\begin{claim} \label{clm:union}
Every separator that $T_1$ and $T_2$ share is also a separator of $T_1 \cup T_2$.
\end{claim}
\begin{proof}
Let $u$ be a separator of both $T_1$ and $T_2$. This implies that in $T_1$ the weight of every connected component in $T_1-u$ is at most $w_1(T_1)/2$ and in $T_2$  the weight of every connected component in $T_2-u$ is at most $w_2(T_2)/2$. Hence, in $T_1 \cup T_2$ the weight of every connected component in $T_1 \cup T_2 -u $ is at most  $w_1(T_1)/2 +w_2(T_2)/2 = w_{1+2}(T_1\cup T_2)/2 $. Thus, $u$ is also a separator of $T_1 \cup T_2$.
\end{proof}

The previous three claims establish the proof of Proposition \ref{prop:median_facts}. First Claim \ref{clm:sep} shows that the set of medians and separators coincide. Then, Claim \ref{clm:sep} Claim \ref{clm:union} prove the two statements of the proposition respectively. 

\subsection{Lower Bounds in Non-Tree Metrics} \label{sec:non-tree}

In some sense tree metrics are the largest class of metrics for which the optimal solution is always a Nash equilibrium for $\alpha \leq \frac 12$. 
The next example demonstrates that even when the distance metric is a simple cycle the PoS can be as high as $\frac 43$ for $\alpha = \frac 12$.
In the following section, we give a family of more involved
constructions that converge to the asymptotically tight lower 
bound of $2$ on the price of stability.

\begin{example}
Consider a metric which is a cycle of size $3k+1$ for some integer $k\geq 1$. Let $A$,$B$,$C$ be three strategies in this strategy space such that $d(A,B)=k$, $d(A,C)=k$ and $d(B,C)=k+1$. Consider an instance where a node with a preferred strategy $A$ is connected to a node with preferred strategy $B$ and to another node with a preferred strategy of $C$. Also, assume that the node with preferred strategy $B$ is part of a clique of size $3k$ in which all nodes prefer strategy $B$. Similarly, the node with preferred strategy $C$ is part of a clique of size $3k$ in which all nodes prefer strategy $C$.

Consider the following equilibrium in which the nodes in both cliques play their preferred strategies. Then, the central node should play its preferred strategy. To see why, note that for playing strategy $A$ its cost is $(1/2)2k=k$. On the other hand, its cost for playing any other strategy $x$ which is between $A$ and $B$ (including $B$) on the cycle is 
\begin{align*}
\frac{1}{2}( d(x,A) + d(x,B) + d(x,C)) = \frac{1}{2} \Big(d(A,B)+min\{d(x,A)+d(A,C),  d(x,B)+d(B,C)\} \Big).
\end{align*}
which is greater than $k$. Similarly one can show that the central player prefers strategy $A$ over any strategy $x$. The cost of the Nash equilibrium is $2k$. Note that this is the best Nash equilibrium since the cost of any solution in which some of the nodes in a clique play a strategy different than their preferred strategy is at least $3k$. In the optimal solution the central node should play strategy $B$ (or $C$) for a total cost of $(1/2)k + 2\cdot (1/2)(k+1) = (3/2)k+1$. Thus we have that the price of stability approaches $4/3$ as $k$ approaches infinity.
\end{example}

Note that this lower bound of $4/3$ is achieved on an instance in
which the lowest-cost Nash equilibrium and the socially optimal solution
differ only in the strategy choice of a single player.
We now show that in such cases, where the difference between these 
two solutions consists of the decision of just a single player,
$4/3$ is the maximum possible price of stability for $\alpha=\frac 1 2$. More
generally we show that $\dfrac{2}{2-\alpha}$  is the maximum possible price of stability for $\alpha <\frac 1 2$.

By the definition of the model,
a player's strategy only affects its cost and the cost of its neighbors. 
Recall that we denote this part of the social cost by $\pc_i(\ao)$:
$\pc_i(\ao)=\alpha \cdot d(\io_i,\ao_i)+2(1-\alpha) \cdot\sum_{j \in N(i)} d(\ao_i,\ao_j)$. 
We now prove the following claim.

\begin{claim}
Let $\alpha \leq \frac 12$. Fix an optimal solution $\oo$ which is not a Nash equilibrium and let player $i$ be a player that can reduce its cost by playing $\no_i$. Then $\dfrac{\pc_i(\no_i,\oo_{-i})}{\pc_i(\oo)} < \dfrac{2}{2-\alpha}$.
\end{claim}
\begin{proof}
Since $\no_i$ is player $i$'s best response, then $\alpha \cdot d(\io_i,\no_i)+(1-\alpha)\sum_{j \in N(i)} d(\no_i,\oo_j) <  \co_i(\oo)$. By rearranging the terms we get that $(1-\alpha) \sum_{j \in N(i)} d(\no_i,\oo_j) <  \co_i(\oo) - \alpha \cdot d(\io_i,\no_i)$. 
This in turn implies that
\begin{align*}
\pc_i(\no_i,\oo_{-i}) &= \alpha \cdot d(\io_i,\no_i)+2(1-\alpha) \sum_{j \in N(i)} d(\no_i,\oo_j) 
&< \alpha \cdot d(\io_i,\no_i)+ 2\big(\co_i(\oo) - \alpha \cdot d(\io_i,\no_i)\big)\\
 &= 2\co_i(\oo) - \alpha \cdot d(\io_i,\no_i).
\end{align*}
Thus, we have that
 $$ \dfrac{\pc_i(\no_i,\oo_{-i})}{\pc_i(\oo)} < \dfrac{2\co_i(\oo) - \alpha \cdot d(\io_i,\no_i)}{\co_i(\oo) + (1-\alpha) \cdot\sum_{j \in N(i)} d(\oo_i,\oo_j)}.$$
If $\sum_{j \in N(i)} d(\oo_i,\oo_j) \geq  \co_i(\oo)$ then $\pc_i(y) \geq \co_i(\oo) + (1-\alpha) \co_i(\oo) = (2-\alpha)\co_i(\oo)$ and the claim follows. Else, we show that $d(\io_i,\no_i) > d(\io_i,\oo_i) - \sum_{j \in N(i)} d(\oo_i,\oo_j)$ in Lemma \ref{lem:int_cost_bound} below; this in turn implies that
\begin{align*}
\pc_i(\no_i,\oo_{-i}) &< 2\co_i(\oo) - \alpha \cdot d(\io_i,\no_i) \leq 2\co_i(\oo) - \alpha \big(d(\io_i,\oo_i) - \sum_{j \in N(i)} d(\oo_i,\oo_j) \big) \\
&=\co_i(\oo) + \sum_{j \in N(i)} d(\oo_i,\oo_j).
\end{align*}

This brings us to the following bound:
\begin{align*}
\dfrac{\pc_i(\no_i,\oo_{-i})}{\pc_i(\oo)} &< \dfrac{\co_i(\oo) + \sum_{j \in N(i)} d(\oo_i,\oo_j)}{\co_i(\oo) + (1-\alpha) \cdot\sum_{j \in N(i)} d(\oo_i,\oo_j)} 
&= 1+ \dfrac{\alpha \sum_{j \in N(i)} d(\oo_i,\oo_j)}{\co_i(\oo) + (1-\alpha) \cdot\sum_{j \in N(i)} d(\oo_i,\oo_j)}.
\end{align*}
Recall that by our assumption $\co_i(\oo) > \sum_{j \in N(i)} d(\oo_i,\oo_j)$, this implies that $ {\co_i(\oo) + (1-\alpha) \cdot\sum_{j \in N(i)} d(\oo_i,\oo_j)} > (2-\alpha)\sum_{j \in N(i)} d(\oo_i,\oo_j)$ and the claim follows.
\end{proof}

\begin{lemma} \label{lem:int_cost_bound}
Let $\alpha \leq \frac 12$. Fix an optimal solution $\oo$ which is not a Nash equilibrium and let player $i$ be a player that can reduce its cost by playing $\no_i$. Then: $d(\io_i,\no_i) > d(\io_i,\oo_i) - \sum_{j \in N(i)} d(\oo_i,\oo_j) $.
\end{lemma}
\begin{proof}
Note the following: by the triangle inequality for any player $j$ it holds that: $d(\no_i,\oo_j) \geq d(\io_i,\oo_j) -d(\io_i,\no_i) $ and $d(\io_i,\oo_j) \geq d(\io_i,\oo_i) - d(\oo_i,\oo_j)$. By combining the two together we have that $ d(\no_i,\oo_j) \geq  d(\io_i,\oo_i) - d(\oo_i,\oo_j) -d(\io_i,\no_i) $. This gives us the following lower bound on $\co_i(\no_i,\oo_{-i})$:
\begin{align*}
\co_i(\no_i,\oo_{-i}) &= \alpha \cdot d(\io_i,\no_i) + (1-\alpha) \sum_{j \in N(i)} d(\no_i,\oo_j) \\
&\geq \alpha \cdot d(\io_i,\no_i) + (1-\alpha) \sum_{j \in N(i)} \Big(  d(\io_i,\oo_i) - d(\oo_i,\oo_j) -d(\io_i,\no_i) \Big) \\
&= \alpha \cdot d(\io_i,\no_i) + (1-\alpha) \cdot |N(i)|\cdot \big(d(\io_i,\oo_i) - d(\io_i,\no_i)\big) - (1-\alpha) \sum_{j \in N(i)} d(\oo_i,\oo_j).  
\end{align*}
Since $\no_i$ minimizes player $i$'s cost it has to be the case that: $\co_i(\no_i,\oo_{-i})  <  \co_i(\oo)$. Thus the following inequality holds:
\begin{align*}
\alpha & \cdot d(\io_i,\no_i) + (1-\alpha) |N(i)|\cdot \big(d(\io_i,\oo_i) - d(\io_i,\no_i) \big) - (1-\alpha) \sum_{j \in N(i)} d(\oo_i,\oo_j)   \\
&<  \alpha \cdot d(\io_i,\oo_i) + (1-\alpha) \sum_{j \in N(i)} d(\oo_i,\oo_j).
\end{align*}

After some rearranging we get that:
\begin{align*}
d(\io_i,\no_i)  > d(\io_i,\oo_i) - \dfrac{2(1-\alpha)}{ (1-\alpha) \cdot |N(i)| -\alpha}\sum_{j \in N(i)} d(\oo_i,\oo_j)
\end{align*}
which implies that the claim holds whenever 
$\dfrac{2(1-\alpha)}{ (1-\alpha) \cdot |N(i)| -\alpha} \leq 1$.
For $\alpha \leq \frac 1 2$, this later bound occurs for $|N(i)| \geq 3$.

The case of $|N(i)|\leq 2$ is handled separately and requires we use the assumption that $\oo$ is an optimal solution. Denote $i$'s neighbors by $j$ and $k$. Then:
\begin{align*}
\alpha \cdot d(\io_i,\oo_j) + 2(1-\alpha) \cdot d(\oo_j,\oo_k) \geq \alpha \cdot d(\io_i,\oo_i) + 2(1-\alpha) \big(d(\oo_i,\oo_j)+d(\oo_i,\oo_k)\big).
\end{align*}
By the triangle inequality the previous inequality implies that $d(\io_i,\oo_j) \geq d(\io_i,\oo_i) $. When combining this with the fact that $d(\no_i,\oo_j) \geq d(\io_i,\oo_j) -d(\io_i,\no_i) $
 we get that $d(\no_i,\oo_j) \geq d(\io_i,\oo_i) -d(\io_i,\no_i) $; similarly we get for $k$ that $d(\no_i,\oo_k) \geq d(\io_i,\oo_i) -d(\io_i,\no_i) $. 
Therefore,
\begin{align*}
\co_i(\no_i,\oo_{-i}) &= \alpha \cdot d(\io_i,\no_i) + (1-\alpha)\big(d(\no_i,\oo_j) + d(\no_i,\oo_k) \big) \\
&\geq \alpha \cdot d(\io_i,\no_i) + 2(1-\alpha)\big( d(\io_i,\oo_i) -d(\io_i,\no_i)\big).
\end{align*}
and since $\no_i$ is player $i$'s best response it has to be the case that:
\begin{align*}
\alpha \cdot d(\io_i,\no_i) + 2(1-\alpha)\big( d(\io_i,\oo_i) -d(\io_i,\no_i)\big) < \alpha \cdot d(\io_i,\oo_i) + (1-\alpha) \big( d(\oo_i,\oo_j) + d(\oo_i,\oo_k) \big).
\end{align*}
After some rearranging we get that:
\begin{align*}
 (2-3\alpha)d(\io_i,\oo_i) -(1-\alpha)\big( d(\oo_i,\oo_j) + d(\oo_i,\oo_k)\big) < (2-3\alpha)d(\io_i,\no_i).
\end{align*}
By dividing both sides of the inequality by $ (2-3\alpha)$ we get that $d(\io_i,\no_i) > d(\io_i,\oo_i) - \sum_{j \in N(i)} d(\oo_i,\oo_j)$ holds whenever $\frac{1-\alpha}{2-3\alpha} \leq 1$. This completes the proof since for $\alpha \leq \frac 12$ it is always the case that $\frac{1-\alpha}{2-3\alpha} \leq 1$.
\end{proof}

\section{Lower Bounds on the Price of Stability}
At the end of the previous section, we saw that even in very
simple non-tree metrics, the price of stability can be greater than $1$.
We now give a set of stronger lower bounds, using a more involved
family of constructions.  First we give an asymptotically tight
lower bound of $2$ when $\alpha = \frac 12$, and then we adapt
this construction to give non-trivial lower bounds for all $0 < \alpha < \frac 12$.

\subsection{Price of stability for $\alpha=\frac 12$} \label{sec:alpha12}
The following example illustrates that the PoS for $\alpha=\frac 12$ can be arbitrarily close to $2$. The network we consider is composed of a path of $n$ nodes and two cliques of size $n^2$ connected to each of the endpoints
of the path. 
We assume that the preferred strategy of node $i$ on the path is $s_i$, 
the preferred strategy of all nodes in the leftmost clique is $s_0$,
and the preferred strategy of all nodes in the rightmost clique is $s_{n+1}$. 
The following is a sketch of the network:
\[  \xygraph{ !{<0cm,0cm>;<1cm,0cm>:<0cm,1cm>::} !{(0,0) }*+{\bigcirc_{0}}="s0" !{(6,0) }*+{\bigcirc_{{n+1}}}="sn1" !{(2,0) }*+{\bullet{1}}="s1" !{(4,0) }*+{\bullet{n}}="sn" !{(3,0) }*+{......}="dummy"  "s0"-"s1" "s1"-"dummy" "dummy"-"sn" "sn"-"sn1"  }  \]

Since all the $s_i$'s are distinct we use them also as names for the 
different possible strategies.
We define the following distance metric on these strategies: 
for $i > j$, we have $d(s_i,s_j)=1+(i-j-1)\epsilon$. 
(When $i < j$, we simply use $d(s_i,s_j) = d(s_j,s_i)$.)
In Claim \ref{clm:unique} below we show that the best Nash equilibrium is the one in which all players play their preferred strategies. The cost of this equilibrium is $c(s) = \frac{1}{2} \cdot 2\sum_{i=0}^{n} d(s_i,s_{i+1}) = n+1$. On the other hand, consider the assignment in which for some node $i$ all the nodes up till node $i$ choose strategy $\io_0$ and all the nodes from node $i+1$ choose strategy $\io_{n+1}$. The cost of such assignment is $\frac 12(n+2+O(\epsilon)) $; therefore as $n$ goes to infinity and $\epsilon$ to zero the PoS goes to $2$.

\begin{claim} \label{clm:unique}
In the previously defined instance the best Nash equilibrium is for each player to play its preferred strategy.
\end{claim}
\begin{proof}
We show that the cost of any other equilibrium is at least $\frac{1}{2} n^2$. Consider an equilibrium in which there exists at least one node $i$ that plays strategy $s_j$ such that $j < i$. By the following lemma (which we prove below) this implies that node $i$'s neighbors play strategies $s_a$ and $s_b $ such that $a,b < i$ or  $a,b > i$.
\begin{lemma} \label{lem:best}
Let $\io_a$ and $\io_b$ be the strategies played by player $i$'s neighbors such that $a \leq b$. If $a \leq i \leq b$, then player $i$'s best response is to play strategy $\io_i$. 
\end{lemma}
Observe that for $i$'s best response to be strategy $s_j$ it clearly has to be the case that $a,b < i$. 
By applying Lemma~\ref{lem:best} repeatedly,
we get that in this equilibrium, all nodes $k>i$ play 
strategies $k'$ such that $k'<k$. 
This includes the $n^{th}$ node of the path, implying that its right neighbor which belongs to the right clique plays strategy $s_{k'}$ such that $k' < n+1$. 
The cost incurred by the nodes in the clique in any such equilibrium 
is at least $n^2$: indeed, if $r$ nodes play a strategy different than 
their preferred one, they pay a cost of at least $\frac{1}{2}r$ 
and the remaining $n^2-r$ pay a cost of at least $n^2-r$ 
for the edges connecting them to one of the $r$ nodes playing 
a strategy different than its preferred strategy. 
To complete the proof, one can use an analogous argument 
for the case in which there exists an equilibrium in which there is a node $i$ playing strategy
$s_j$ such that $j>i$. 
\end{proof}

\noindent \textbf{Proof of Lemma \ref{lem:best}:}
We first note that by the definition of the metric it is never in $i$'s best interest to play a strategy $\io_j$ such that $j\neq i, a, b$:
playing such a strategy has cost $3+O(\epsilon)$ whereas the cost of playing the preferred strategy is $2+O(\epsilon)$. 

Observe that player $i$ prefers to play strategy $s_i$ over strategy $s_a \neq s_i$ whenever $d(s_i,s_a) + d(s_i,s_b) <  d(s_i,s_a) + d(s_a,s_b)$ implying that $d(s_i,s_b) <  d(s_a,s_b)$. This conditions holds according to our assumptions since $1+(b-i-1)\epsilon=d(s_i,s_b) <  d(s_a,s_b) =1+(b-a-1)\epsilon$. For the same reason, player $i$ prefers strategy $s_i$ over $s_b \neq s_i$ since $1+(i-a-1)=d(s_i,s_a) <  d(s_a,s_b) = 1+ (b-a-1)$ under the 
lemma's assumptions.
$\qed$

\subsection{Extension for $\alpha<\frac 12$}
We extend the construction in Section \ref{sec:alpha12} to $0< \alpha<\frac 12$ by defining the following metric:
for $i > j$, 
let
$d(s_i,s_j)=1+(i-j-1) \left(\dfrac{1-2\alpha}{1-\alpha}(1+\epsilon)\right)$. 
We consider the same family of instances defined in Section \ref{sec:alpha12} except for the fact that we increase the size of the cliques to $n^2 / \alpha$.

Next, we show that Lemma~\ref{lem:best} also holds for this newly defined family of instances. This fact together with the observation that the proof of Claim \ref{clm:unique} carries over with only minor modifications, imply that in the best Nash equilibrium of the previously defined family of instances all players play their preferred strategies.

\begin{lemma} \label{lem:generlized_best}
Let $\io_a$ and $\io_b$ be the strategies played by player $i$'s neighbors such that $a \leq b$. If $a \leq i \leq b$, then player $i$'s best response is to play strategy $\io_i$. 
\end{lemma}
\begin{proof}
For proving this claim it will be easier to use an equivalent distance function which is: $d(s_i,s_j)=1+(|i-j|-1) \left(\dfrac{1-2\alpha}{1-\alpha}(1+\epsilon)\right)$. Also, we only present the proof for the case that $a < i < b$, the proof for the rest of the cases is very similar. We first show that player $i$ prefers to play strategy $\io_i$ over playing any strategy $s_j$ such that $j \neq a ,b$. 
The cost of player $i$ for playing $s_i$ is:
\begin{align*}
&(1-\alpha)\left(2+ (|i-a|-1+|b-i|-1)\cdot \dfrac{1-2\alpha}{1-\alpha}(1+\epsilon)  \right)  \\
&= 2(1-\alpha)+  (|b-a|-2)\cdot (1-2\alpha)(1+\epsilon)\\
&\leq 1+  (|b-a|-1)\cdot (1-2\alpha)(1+\epsilon).
\end{align*} 
The cost of playing strategy $\io_j$ such that $j \neq a ,b$ is:
\begin{align*}
& \alpha(1+ (|i-j|-1)\dfrac{1-2\alpha}{1-\alpha}(1+\epsilon)) + (1-\alpha)\left(2+ (|j-a|-1 + |j-b|-1)\cdot \dfrac{1-2\alpha}{1-\alpha}(1+\epsilon)  \right)  \\
&\geq 2-\alpha+ (|b-a|-2)\cdot (1-2\alpha)(1+\epsilon).  
\end{align*} 
The last transition is due to the fact that $|j-a|+ |j-b| \geq |b-a|$ . Thus we conclude that player $i$'s best response can only be $\io_i,\io_a$ or $\io_b$.
Next we consider strategies $\io_a$ and $\io_b$. By writing the cost of playing each one of these strategies it is easy to see that these costs are
greater than the costs for playing $\io_i$.

The cost of playing strategy $\io_a \neq \io_i$ is:
\begin{align*}
& \alpha(1+ (|i-a|-1)\dfrac{1-2\alpha}{1-\alpha}(1+\epsilon)) + (1-\alpha)\left(1+ (|b-a|-1)\cdot \dfrac{1-2\alpha}{1-\alpha}(1+\epsilon)  \right)  \\
&= 1+ (\dfrac{\alpha}{1-\alpha}(|i-a|-1)+|b-a|-1)\cdot (1-2\alpha)(1+\epsilon). 
\end{align*} 

The cost of playing strategy $\io_b \neq s_i$ is:
\begin{align*}
&\alpha(1+ (|b-i|-1)\dfrac{1-2\alpha}{1-\alpha}(1+\epsilon)) + (1-\alpha)\left(1+ (|b-a|-1)\cdot \dfrac{1-2\alpha}{1-\alpha}(1+\epsilon)  \right)  \\
&= 1+ (\dfrac{\alpha}{1-\alpha}(|b-i|-1)+|b-a|-1)\cdot (1-2\alpha)(1+\epsilon).
\end{align*} 
This conclude the proof as we have shown that under the assumptions of the claim, player $i$'s best response it to play its preferred strategy.
\end{proof}

To get a lower bound on the PoS, we would like to simulate the technique
we used in the proof for $\alpha = \frac 1 2$ to compare between the cost 
of two solutions:
(i) a solution in which there is exactly one edge such 
that its two endpoints play different strategies, and (ii)
the best Nash equilibrium. We refer to the first solution as a {\em bi-consensus solution}. 
The cost of the best bi-consensus solution is an upper bound 
on the optimal solution, and hence
computing the ratio between the best Nash equilibrium and best bi-consensus solution gives a lower bound on the PoS achieved by instances defined above.

Observe that in the best bi-consensus solution nodes $i \in [1\dots \lfloor n/2 \rfloor]$ play strategy $\io_0$ and nodes $i \in [\lceil n/2 \rceil \dots n]$ play strategy $\io_{n+1}$.
The cost of this solution $b$ is the following:
\begin{align*}
\co(b) &=\alpha \Big( \sum_{i=1}^{\lfloor n/2 \rfloor} (1+(i-1) (\dfrac{1-2\alpha}{1-\alpha}(1+\epsilon))) 
	~+\sum_{i=\lfloor n/2 \rfloor+1}^{n} (1+(n-i) (\dfrac{1-2\alpha}{1-\alpha}(1+\epsilon))) \Big) \\
	&~~~+2(1-\alpha) (1+(n+1-0-1) (\dfrac{1-2\alpha}{1-\alpha}(1+\epsilon)))\\
	&\leq \alpha \cdot n + \alpha \frac{1}{4} (n-1)^2 \cdot \dfrac{1-2\alpha}{1-\alpha}(1+\epsilon) + 2(1-\alpha) + 2n(1-2\alpha)(1+\epsilon) .
\end{align*}	

Where the last transition is due to the fact that:
\begin{align*}
\sum_{i=1}^{\lfloor n/2 \rfloor} (i-1) + \sum_{i=\lfloor n/2 \rfloor+1}^{n} (n-i) 
&\leq \sum_{i=1}^{\lfloor n/2 \rfloor-1} i + \sum_{i=1}^{\lfloor n/2 \rfloor} i
= (\lfloor n/2 \rfloor-1) \cdot \lfloor n/2 \rfloor + \lfloor n/2 \rfloor \\
&= \lfloor n/2 \rfloor^2
= \frac{1}{4}(n-1)^2.
\end{align*}

The cost of the best Nash equilibrium $\no$ is simply $\co(\no) = 2(1-\alpha)(n+1)$.
Interestingly, once we pick $\alpha < \frac 12$ the maximum PoS for 
this example is obtained for an intermediate value of $n$. 
By taking the first derivative of 
the function $\frac{\co(\no)}{\co(b)}$ with respect to $n$ and comparing it to $0$, we get that the maximum PoS is achieved for $n=\lceil \frac{1-2 \alpha-2 \sqrt{2-7 \alpha+6 \alpha^2}}{-1+2 \alpha} \rceil$ or $n=\lfloor  \frac{1-2 \alpha-2 \sqrt{2-7 \alpha+6 \alpha^2}}{-1+2 \alpha} \rfloor$. 

In Figure \ref{fig:posless} we plot the lower bound on the PoS 
that can be achieved by this example (solid line). 
As one might can expect, as $\alpha$ approaches $\frac 12$ 
the PoS approaches $2$.  For comparison, we also plot (via the dashed line)
the lower bound of $\frac{2}{2-\alpha}$ on the PoS 
that we computed in Section \ref{sec:non-tree} for instances in which the best Nash equilibrium differs from an optimal solution only in the strategy played by a single player. 
Interestingly, each of the two constructions offers a better lower bound 
on the PoS for a different interval of $\alpha$. 
We cannot rule out that the maximum of these two constructions
could match the best achievable upper bound on the PoS for all $\alpha$,
but there may also be other constructions that can achieve higher lower
bounds on the PoS for some ranges of $\alpha$.
\begin{figure}[htb]
\begin{center}
\includegraphics[width=2.5in]{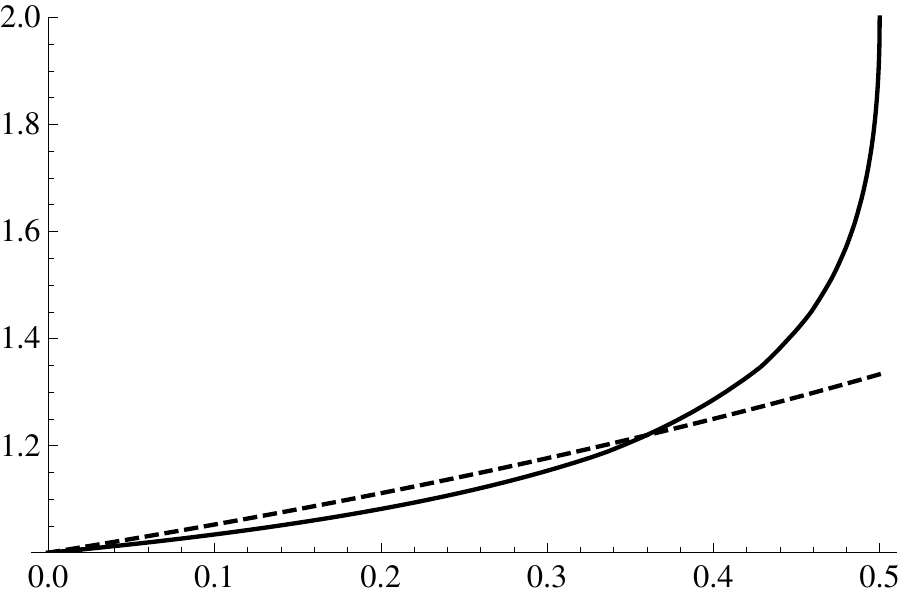}
\end{center}
\caption{The PoS achievable by a path (solid) and by a single strategic node (dashed).}
\label{fig:posless}
\end{figure}

\newcommand{\fixed}{F}    
\newcommand{\str}{S}    

\section{The Anchored Preference Game}
In this final section, we consider the following generalization
of discrete preference games, with $\alpha=\frac 12$.
We assume that nodes are partitioned into two types: $\fixed$ and $\str$.
Nodes in $\fixed$ are {\em fixed nodes} that always play their preferred strategy.
The cost that a node $i \in \fixed$ incurs for playing strategy $\io_i$ is $0$.
The nodes in $\str$ are {\em strategic nodes}
that have no preferred strategy, so the cost of a node $i \in \str$ is purely the
term $\sum_{j \in N(i)} d(\ao_i,\ao_j)$. 
We use $\fixed(i)$ and $\str(i)$ to denote
the fixed and strategic neighbors of node $i$ respectively.
Only the strategic nodes choose strategies, and only they take
part in the definition of the equilibrium.
The social cost for this model is:
$\displaystyle{
\co(\ao)= \sum_{\substack{(i,j) \in E; \\  i \in \str; j \in \fixed}} d(\ao_i,\io_j)+2\sum_{\substack{(i,j) \in E; \\ i,j \in \str}}  d(\ao_i,\ao_j).}$

As noted in the introduction, this model, which we call
an {\em anchored preference game}, 
generalizes the model from the previous sections 
because of the following reduction: 
given an instance of a discrete preference game,
we can take each node $i$ in the instance and make it a strategic
node in the generalized instance by eliminating its preferred
strategy $s_i$, and adding a new fixed node $i'$ to the instance 
that has preferred strategy $s_i$ and is connected only to node $i$
by an edge $(i,i')$.
In this way $i'$, which is non-strategic, plays the role of $i$'s
preferred strategy.

Given that discrete preference games are (due to the reduction)
a special case of anchored preference games with at most one fixed 
neighbor per node, it becomes natural to study the price of stability
of anchored preference games parametrized by $k$, the maximum number
of fixed nodes adjacent to any strategic nodes.
In the next claim we generalize our result on tree metrics by showing that the price of
stability for anchored preference games with tree metrics is $1$
provided $k \leq 2$. 
The proof is essentially
a generalization of Theorem \ref{thm:tree_metric}.

\begin{claim} \label{claim-twoFixed} 
If the distance function is a tree metric and $k\leq 2$, then the optimal solution is also a Nash equilibrium. 
\end{claim}
\begin{proof}
Similar to the proof of Theorem \ref{thm:tree_metric}, we define $\br_i(\ao)$ and $\sbr_i(\ao) $ to be the strategies of player $i$ that minimize $c_i(\ao) = \sum_{j \in \fixed(i)} d(\ao_i,\io_j) + \sum_{j \in \str(i)} d(\ao_i,\ao_j)$ and ${\pc}_i(\ao) = \sum_{j \in \fixed(i)} d(\ao_i,\io_j) + 2\sum_{j \in \str(i)} d(\ao_i,\ao_j)$ respectively. Our goal now is to show that the set of best responses $\br(i)$ and the set of local improvements $\sbr(i)$ always intersect. By Claim \ref{clm:best_response} this implies that the PoS is always $1$. 

Observe that $\br_i(\ao) \cap \sbr_i(\ao) \neq \emptyset$ when the number of fixed neighbors node $i$ has is $0$ (by definition) and when the number of fixed neighbors node $i$ has is $1$ (by Lemma \ref{lem:intersect}). Thus, we can assume that player $i$ has exactly 2 fixed neighbors. We denote their preferred strategies by $\io_1$ and $\io_2$. Just as in the proof of Theorem \ref{thm:tree_metric} we show that $\br_i(\ao)$ and $\sbr_i(\ao)$ coincide with the set of medians for some trees we define next. For this purpose we define the tree $T_{i,\ao}(q_1,q_2,r)$:

\begin{definition}
Given a metric tree $T$, a strategy vector $\ao$ and a strategic player $i$ we denote by $T_{i,\ao}(q_1,q_2,r)$ the tree with the same nodes and edges as $T$ and the following node weights:
\begin{align*}
w(v) = r\cdot |\{j\in \str(i) | \ao_j = v \}| + \sum_{\{j | j\in\{1,2\}, \io_j =v\}} q_j.
\end{align*}
\end{definition}

It is not hard to see that for a player $i$ with two fixed neighbors $\br(i) = M(T_{i,\ao}(1,1,1))$ and  $\sbr(i) = M(T_{i,\ao}(1,1,2))$. Recall that by Claim \ref{clm:sep} the set of medians and separators of trees coincide. Therefore, we should show that the set of separators of $T_{i,\ao}(1,1,1)$ and $T_{i,\ao}(1,1,2)$ intersect. By Claim \ref{clm:diff_by_one} we already have that $T_{i,\ao}(2,1,2)$ and $T_{i,\ao}(1,1,2)$ share a separator $u$. Since separators are invariant to scaling $T_{i,\ao}(1,1,1)$ and $T_{i,\ao}(2,2,2)$ have the same separators. Thus, to complete the proof we need to show that $u$ is a separator of $T_{i,\ao}(2,2,2)$.

Observe that the sum of weights in $T_{i,\ao}(2,1,2)$ is odd; this implies that the size of every connected component of $T_{i,\ao}(2,1,2) - u $ is at most $w(T_{i,\ao}(2,1,2))/2-1/2$. After increasing the weight of a single node by exactly $1$ the weight of each component is at most $w(T_{i,\ao}(2,1,2))/2+1/2 =w(T_{i,\ao}(2,2,2))/2 $. Thus, $u$ is also a separator of $T_{i,\ao}(2,2,2)$ implying that  $T_{i,\ao}(1,1,1)$ and $T_{i,\ao}(1,1,2)$ indeed share a separator.
\end{proof}

We use Claim \ref{claim-twoFixed} 
to show that for tree metrics and $k>2$ 
the PoS is bounded by $\frac{2(k-1)}{k}$. 
The proof is again given in the appendix.
\begin{claim} \label{claim-arb-fixed}
If the distance function is a tree metric and each strategic node has at most $k$ fixed neighbors, then $PoS \leq \frac{2(k-1)}{k}$.
\end{claim}
\begin{proof}
Let $\oo$ be the optimal solution for a graph $G$ with a total 
of $n$ nodes ($|\fixed|+|\str| =n$), and 
let $c_G(y)$ denote the social cost of this solution.
By duplicating fixed nodes if necessary, 
we can assume without loss of generality that each fixed node
has at most one strategic neighbor. Denote by $\no$ the best Nash equilibrium
for this instance.

We construct a new graph $G'$ from $G$ in which each strategic
node will have at most
two fixed neighbors.  For each strategic node $i$ in $G$, we pick 
the $k-2$ fixed neighbors of $i$ with preferred strategies which are closest 
to $\oo_i$ in distance. 
We replace each such fixed neighbor $j$ 
with a clique of size $\lceil 2c_G(\oo) \rceil$ consisting of 
strategic nodes, where the preferred strategy of each node in the clique 
is $\io_j$. We connect node $i$ to just one of the nodes in the clique. 
Since $j$ had $i$ as its only strategic neighbor in $G$,
at the end of this process there is a single node connected to each clique. 

Let $c_{G'}(\cdot)$ be the social cost function for $G'$.
Since each strategic node has at most two fixed neighbors in $G'$,
Claim \ref{claim-twoFixed} implies that there exists 
an optimal solution $y'$
for $G'$ which is also a Nash equilibrium in $G'$. 
Observe that in any optimal solution of $G'$ all the nodes in 
the new cliques we created play their preferred strategies, since the cost of any other strategy vector is at least $\lceil 2c_G(\oo) \rceil$ and the cost of the optimal solution in $G'$ cannot be larger than twice the cost of the optimal solution in $G$. 

We now define two further solutions. Let 
$\oo'_G$ be a strategy vector for $G$ 
in which the nodes common to $G$
and $G'$ play their strategies in $\oo'$, and for each fixed node $j$
replaced by a clique, $j$ plays its preferred strategy.
(Above we argued that all the nodes in the clique replacing $j$ 
will play this same preferred strategy in $G'$.) We observe that $\oo'_G$ is a well-defined solution to the anchored
preference game on $G$, since all fixed nodes play their preferred strategy,
and $\oo'_G$ is an equilibrium in $G$, since $\oo'$ is an equilibrium
in $G'$. Therefore we have that $c_G(\no) \leq c_G(\oo'_G)$. Now, since $c_G(\ao) \leq  c_{G'}(\ao)$ and by construction $\oo'$ and $\oo'_G$ are practically the same, we have that $c_G(\oo'_G) \leq c_{G'}(\oo')$. 

Let $\oo_{G'}$ denote the strategy vector for $G'$
that agrees with $y$ on the nodes of $G$, and in which the nodes of
each clique in $G'$ play the preferred strategy of the fixed node $j$ they 
replaced. Since $\oo'$ is an optimal solution for $G'$ we have that $c_{G'}(\oo') \leq c_{G'}(\oo_{G'})$. Thus, we have that $c_G(\no)\leq c_{G'}(\oo_{G'})$. 

For simplicity of presentation we define $F_i=\sum_{j\in \fixed(i)} d(\oo_i,\io_j)$ to be the cost of player $i$ associated with its fixed neighbors in the graph $G$. Since in $G'$ the cost of the $k-2$ fixed neighbors which are closest to $y_i$ was doubled, it is now at most $(1+\frac{k-2}{k}) F_i$.
\begin{equation*}
PoS(G) = \dfrac{c_G(x)}{c_G(y)} \leq\dfrac{c_{G'}(\oo_{G'})}{c_G(y)} \leq \dfrac{\sum_{i \in \str} \sum_{j \in \str(i) }d(\oo_i, \oo_j) + (1+\frac{k-2}{k})\sum_{i\in \str} F_i}{\sum_{i \in \str} \sum_{j \in \str(i) }d(\oo_i,\oo_j) + \sum_{i \in \str} F_i} \leq \dfrac{2k-2}{k}.
\end{equation*}
\end{proof}

To see that the bound for $k>2$ is tight, consider a star in which the central node $i$ is connected to $k$ fixed nodes that prefer strategy $A$. $i$ is also connected to $k-1$ strategic nodes that are connected to one another 
(forming a clique of size $k-1$) and each one is connected to $k$ fixed nodes that prefer strategy $B$. Observe that in the best Nash equilibrium node $i$ plays strategy $A$ and the rest of the strategic nodes play strategy $B$. The social cost of this equilibrium is $2(k-1)$. However, in the optimal solution node $i$ also plays strategy $B$ which reduces the social cost to $k$.

\subsubsection*{Acknowledgements} We thank Shahar Dobzinski for valuable comments.

\end{document}